\makeatletter \let\@twosidetrue\@twosidefalse \let\@mparswitchtrue\@mparswitchfalse \makeatother
\documentclass[11pt,envcountsame]{article}
\pdfoutput=1

\usepackage{amsmath}
\usepackage{amssymb}
\usepackage{amsthm}
\usepackage[marginparwidth=2.4cm,margin=2.8cm]{geometry}

\usepackage{verbatim}
\usepackage{graphicx}
\usepackage{xspace}

 \newtheorem{theorem}{Theorem}
 \newtheorem{lemma}[theorem]{Lemma}
 \newtheorem{corollary}[theorem]{Corollary}
 \newtheorem{observation}[theorem]{Observation}

\newcommand{\bigO}{\mathcal{O}}
\usepackage{todonotes}

\title{Ants: Mobile Finite State Machines}

\author{
Yuval Emek\thanks{ETH Zurich, Switzerland}
\and
Tobias Langner\thanks{ETH Zurich, Switzerland}
\and
Jara Uitto\thanks{ETH Zurich, Switzerland}
\and
Roger Wattenhofer\thanks{Microsoft Research, Redmond, WA and ETH Zurich,
Switzerland}
}

\date{}

\newcommand\atype[1]{\textsf{#1}\xspace}
\newcommand\guide{\atype{Guide}}
\newcommand\guides{\atype{Guides}}
\newcommand\mguide{\atype{MovingGuide}}
\newcommand\mguides{\atype{MovingGuides}}
\newcommand\nguide{\atype{NewGuide}}
\newcommand\nguides{\atype{NewGuides}}
\newcommand\explorer{\atype{Explorer}}
\newcommand\explorers{\atype{Explorers}}
\newcommand\nexplorer{\atype{NewExplorer}}
\newcommand\nexplorers{\atype{NewExplorers}}
\newcommand\mexplorer{\atype{MovingExplorer}}
\newcommand\mexplorers{\atype{MovingExplorers}}
\newcommand\stime{t^s}
\newcommand\ftime{t^f}

\newcommand\RectSearch{\textsf{RectSearch}\xspace}

\newcommand\GeomSearch{\textsf{GeomSearch}\xspace}
\newcommand\HybridSearch{\textsf{HybridSearch}\xspace}
\newcommand\PSTA{\textsf{ParallelTeamAssignment}\xspace}
\makeatletter
\def\whp{w.h.p\@ifnextchar.{}{.\ }}
\makeatother
\newcommand\FS{\textsf{FastSpread}\xspace}

\begin{document}

\begin{titlepage}

	\maketitle

	\begin{abstract}
		Consider the \emph{Ants Nearby Treasure Search (ANTS)} problem introduced by
		Feinerman, Korman, Lotker, and Sereni (PODC 2012), where $n$ mobile
		agents, initially placed at the origin of an infinite grid, collaboratively
		search for an adversarially hidden treasure.
		In this paper, the model of Feinerman et al.\ is adapted such that the agents
		are controlled by a (randomized) finite state machine: they possess a
		constant-size memory and are able to communicate with each other through
		constant-size messages.
		Despite the restriction to constant-size memory, we show that their
		collaborative performance remains the same by presenting a distributed
		algorithm that matches a lower bound established by Feinerman et al.\ on
		the run-time of any ANTS algorithm.
	\end{abstract}

	\subsection*{Contact Author Information}
	\noindent Tobias Langner \\
	ETH Zurich \\
	Gloriastrasse 35, ETZ G61.4 \\
	8092 Zurich \\
	Switzerland \medskip
	
	\noindent phone: +41 44 63 27730 \\
	\noindent email: \texttt{tobias.langner@tik.ee.ethz.ch}
\end{titlepage}

\section{Introduction}

``They operate without any central control.
Their collective behavior arises from local interactions.''
The last quote is arguably the mantra of distributed computing, however, in
this case, ``they'' are not nodes in a distributed system;
rather, this quote is taken from a biology paper that studies social insect
colonies~\cite{PDG12}.
Understanding the behavior of insect colonies from a distributed computing
perspective will hopefully prove to be a big step for both
disciplines.

In this paper, we study the process of food finding and gathering by ant
colonies from a distributed computing point of view.
Inspired by the model of Feinerman et al.~\cite{Feinerman2012}, we
consider a colony of $n$ ants whose nest is located at the origin of an
infinite grid that collaboratively search for an adversarially hidden food
source.
An ant can move between neighboring grid cells in a single time unit and
can communicate with the ants that share the same grid cell. However, the ant's
navigation and communication capabilities are very limited since its actions
are controlled by a randomized \emph{finite state machine} (FSM) --- refer to
Section~\ref{section:Model} for a formal definition of our model.
Nevertheless, we design a distributed algorithm ensuring that the ants locate
the food source within $\bigO(D + D^2 / n)$ time units \whp,
where $D$ denotes the distance between the food source and the nest.\footnote{
We say that an event occurs \emph{with high probability}, abbreviated by \whp,
if the event occurs with probability at least $1 - n^{-c}$, where $c$ is an
arbitrarily large constant.
}
It is not difficult to show that a matching lower bound holds even under the
assumptions that the ants have unbounded memory (i.e., are controlled by a Turing
machine) and that the ants know the parameter $n$.

\subsection{Related Work}

Our work is strongly inspired by Feinerman et
al.~\cite{Feinerman2012,Feinerman2012disc} who introduce the
aforementioned problem called \emph{ants nearby treasure search (ANTS)} and
study it assuming that the ants (a.k.a.\ \emph{agents}) are controlled by a
Turing machine (with or without space bounds) and that communication is
allowed only in the nest.
They show that if the $n$ agents know a constant approximation of $n$, then
they can find the food source (a.k.a.\ \emph{treasure}) in time $\bigO(D + D^2
/ n)$.
Moreover, Feinerman et al.\ observe a matching lower bound and prove that this
lower bound cannot be matched without some knowledge of $n$.
In contrast to the model studied in
\cite{Feinerman2012,Feinerman2012disc}, the agents in our model can
communicate anywhere on the grid as long as they share the same grid cell.
However, due to their weak control unit (a FSM), their
communication capabilities are very limited even when they do share the same
grid cell (see Section~\ref{section:Model}).
Notice that the stronger computational model assumed by Feinerman et al.\
enables an individual agent in their setting to perform tasks way
beyond the capabilities of a (single) agent in our setting, e.g., list the grid
cells it has already visited or perform spiral searches (that play a major role
in their upper bound).

Distributed computing by finite state machines has been studied in several
different contexts including \emph{population protocols}
\cite{AngluinADFP2006,AspnesR2009} and the recent work of Emek and Wattenhofer
\cite{Emek2013} from which we borrowed the agents communication model (see
Section~\ref{section:Model}).
In that regard, the line of work closest to our paper is probably the one
studying graph exploration by FSM controlled agents, see, e.g.,
\cite{FraigniaudIPPP2005}.

Graph exploration in general is a fundamental problem in computer science.
In the typical case, the goal is for a single agent to visit all nodes in a
given graph.
As an example, the exploration of trees was studied in~\cite{Diks2004}, the
exploration of finite undirected graphs was studied
in~\cite{Panaite1998,Reingold2008}, and the exploration of strongly connected
digraphs was studied in~\cite{Deng1999,Albers2000}.
When a deterministic agent is exploring a graph, memory usage becomes an
issue.
With randomized agents, it is well-known that random walks allow a single
agent to visit all nodes of a finite undirected graph in polynomial time
\cite{Aleliunas1979}.
The speed up gained from using multiple random walks was studied by Alon et
al.~\cite{Alon2008}.
Notice that in an infinite grid, the expected time it takes for a random walk
to reach any designated cell is infinite.

Finding treasures in unknown locations has been previously studied, for example,
in the context of the classic \emph{cow-path} problem.
In the typical setup, the goal is to locate a treasure on a line as quickly as
possible and the performance is measured as a function of the distance to the
treasure.
It has been shown that there is a deterministic algorithm with a competitive
ratio $9$ and that a spiral search algorithm is close to optimal in the
2-dimensional case~\cite{Baeza-Yates1993}.
The study of the cow-path problem was extended to the case of multiple agents
by L\'{o}pez-Ortiz and Sweet~\cite{Lopez-Ortiz2001}.
In their study, the agents are assumed to have unique identifiers, whereas our
agents cannot be distinguished from each other (at least not at the beginning
of the execution).

\subsection{Model}
\label{section:Model}

We consider a variant of \cite{Feinerman2012}'s ANTS problem, where $n$
mobile \emph{agents} are searching for an adversarially hidden \emph{treasure}
in the infinite grid $\mathbb{Z}^{2}$.
The execution progresses in synchronous \emph{rounds}, where each round lasts
$1$ time unit.
At time $0$ (the beginning of the execution), all agents are positioned
in a designated grid cell referred to as the \emph{origin} (say, the cell with
coordinates $(0, 0) \in \mathbb{Z}^{2}$).
We assume that the agents can distinguish between the origin and the other
cells.

The main difference between our variation of the ANTS model and the original
one lies in the agents' computation and communication capabilities.
In both variants, all agents run the same (randomized) protocol, however, 
under the model considered in the current paper, the agents are controlled by
a randomized \emph{finite state machine (FSM)}.
This means that the individual agent has a constant memory and thus, in general,
cannot store its coordinates in $\mathbb{Z}^{2}$.
On the other hand, in contrast to the model considered in
\cite{Feinerman2012}, the communication of our agents is not restricted to
the origin.
Specifically, under our model, an agent $a$ positioned in cell $c \in
\mathbb{Z}^{2}$ can communicate with all other agents positioned in cell $c$ at
the same time.
This communication is quite limited though:
agent $a$ merely senses for each state $q$ of the finite state machine, whether
there exists at least one agent $a' \neq a$ in cell $c$ whose current state is
$q$.
Notice that this communication scheme is a special case of the one-two-many
communication scheme introduced in \cite{Emek2013} with bounding parameter $b
= 1$.

The \emph{distance} between two grid cells $(x, y), (x', y') \in
\mathbb{Z}^{2}$ is defined with respect to the $\ell_{1}$ norm (a.k.a.\
Manhattan distance), that is, $|x - x'| + |y - y'|$.
The two cells are called \emph{neighbors} if the distance between them is $1$.
In each round of the execution, agent $a$ positioned in cell $(x, y) \in
\mathbb{Z}^{2}$ can either move to one of the $4$ neighboring cells $(x,
y + 1), (x, y - 1), (x + 1, y), (x - 1, y)$, or stay put in cell $(x, y)$.
The former $4$ \emph{position transitions} are denoted by the corresponding
cardinal directions $N, S, E, W$, whereas the latter (stationary) position
transition is denoted by $P$ (standing for ``stay put'').

Formally, the agents' protocol is captured by the $3$-tuple
\[
\Pi = \left\langle Q, s_{0}, \delta \right\rangle \, ,
\]
where
$Q$ is the finite set of \emph{states};
$s_{0} \in Q$ is the \emph{initial state}; and
$\delta : Q \times 2^{Q} \rightarrow 2^{Q \times \{ N, S, E, W, P \}}$
is the \emph{transition function}.
At time $0$, all agents are in state $s_{0}$ and positioned in the origin.
Assuming that at time $t \geq 0$, agent $a$ is in state $q \in Q$ and
positioned in cell $c \in \mathbb{Z}^{2}$, the state $q' \in Q$ of $a$ at time
$t + 1$ and its position transition $\tau \in \{ N, S, E, W, P \}$ are chosen
uniformly at random from pairs
$(q', \tau) \in \delta(q, Q_{a})$,
where $Q_{a} \subseteq Q$ contains state $p \in Q$ if and only if there exists
some (at least one) agent $a' \neq a$ such that at time $t$, agent $a'$ is in
state $p$ and positioned in cell $c$.

The goal of the agents is to locate an adversarially hidden
\emph{treasure}, i.e., to bring at least one agent to the cell in which the
treasure is positioned.
The performance measure of our algorithms is their run-time expressed in terms
of the number $n$ of agents and the distance $D$ between the origin and the
treasure.
It is important to point out that neither of the two parameters $n$ and $D$ is
known to the agents (who cannot even store them in their very limited memory).

\section{Parallel Rectangle Search}
\label{sec:RectSearch}

In this section, we introduce the collaborative search strategy \RectSearch
that depends on an \emph{emission scheme}, which divides all agents in the
origin into \emph{search teams} of size five and emits these teams
continuously from the origin until all search teams have been emitted.
We delay the description of our emission scheme until Section~\ref{sec:almost
optimal emission scheme} and describe for now the general search strategy
(without a concrete emission scheme).
Whenever a team is emitted, four agents become \emph{guides} --- one for each cardinal direction --- and the fifth one becomes an \emph{explorer}. Now each guide walks into its respective direction until it hits the first cell that has not been occupied by a guide before (this might later in the execution require to first traverse empty cells until the block of active guides is found and then walk to the end of this block). The explorer follows the north guide and when they hit the first not yet covered cell $(0, d) \in \mathbb Z^2$ for some $d > 0$, the explorer starts a rectangle search by first walking south-west towards the west guide. When it hits a guide, the explorer changes its direction to south-east, then to north-east, and finally to north-west. This way, the explorer traverses all cells in distance 
$d$ from the origin (and in passing also almost all cells in distance $d+1$). 

Whenever the explorer meets a guide on its way, the respective guide moves further outwards to the next empty cell --- hopping over all the other guides on its way if any --- and waits there for the explorer's next appearance. When the explorer has finished its rectangle by reaching the north guide again, it moves north --- together with the north guide --- to the first empty cell and starts another rectangle search there. All other guides have as well reached their target positions in the same distance from the origin and a new search can begin. Figure~\ref{fig:rectangle search} gives an illustration of the process for a single search team. We will now describe the individual aspects of the \RectSearch algorithm in a more precise and formal way.

\begin{figure}[htb]
	\centering
	\includegraphics[scale=0.9]{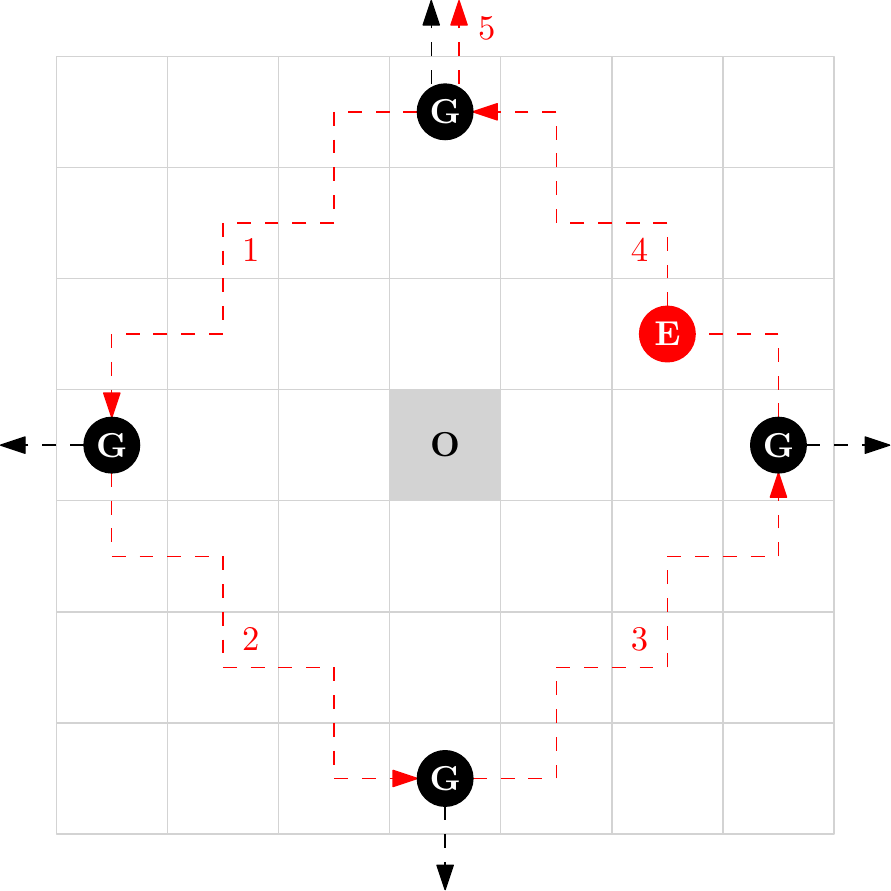}
	\caption{An explorer (E) starts a rectangle search in distance $d$ (here $d = 3$) from the origin $O$ at the north guide, visits all guides (G) in distance $d$ in a counter-clockwise fashion and ends at the north guide. Whenever the explorer meets a guide, the guide moves outwards in its cardinal direction. When the explorer completes its search by arriving again at the north guide, both agents walk outwards together to the next distance to be searched. (The red numbers indicate the order in which the explorer moves.)}
	\label{fig:rectangle search}
\end{figure}

\paragraph{Emission Scheme.}
Initially, all  $n$ agents are located at the origin. Until all agents have been emitted, an \emph{emission scheme} emits new search teams consisting of five agents from the origin satisfying the following two properties: (i)~no two search teams are emitted at the same time $t$ and (ii)~until all search teams have been emitted, the number of search teams emitted until time $t$ is $f_n(t)$ for some \emph{emission function} $f_n$.

Whenever a search team is ready, four of the five agents become \nguides\ --- 
one for each cardinal direction --- and walk outwards in their corresponding
directions, while the fifth one becomes a \nexplorer and follows the north
\nguide (see below for a detailed description of the agent types). 

\paragraph{Agent Types.} In the rest of the paper, we will refer to several different types of agents. Since there is only a constant number of different types, these can be modelled by having individual finite automata for the various types. We use six different types and explain their specific behavior in the following: \guide, \nguide, \mguide, \explorer, \nexplorer, and \mexplorer. Since agents of the three \guide-types are associated with a cardinal direction, we will use the term ``outwards'' to indicate their respective direction. We subsume the types \explorer, \nexplorer, and \mexplorer under the name \emph{exploring agents}.

\paragraph{\nguide.}
A \nguide moves outwards until it hits the first cell containing a \guide. From then on, it continues outwards until it hits a cell that contains neither a \guide nor a \mguide, and stops in this cell, becomes \guide and waits for an \explorer to visit. The \nguides of the first search team stop on the first cell outwards from the origin.

\paragraph{\nexplorer.}
A \nexplorer acts like a \nguide unless that when it hits a cell that contains neither a \guide nor an \mguide, it moves one cell west and becomes an \explorer.

\paragraph{\guide.}
A \guide remains dormant unless its cell is visited by an \explorer whereafter it moves a cell outwards and becomes a \mguide.

\paragraph{\mguide.}
A \mguide moves outwards until it hits a cell that does not contain a \guide. Then it stops there, becomes a \guide again, and waits for an \explorer to visit.

\paragraph{\mexplorer.}
A \mexplorer initially moves north along with ``its'' \mguide. It does so until it hits the first cell that does not contain a \guide. Then it moves one cell west and becomes an \explorer.

\paragraph{\explorer.}
An \explorer does the bulk of the actual search process by moving along the sides of a rectangle using \guides on its way to change direction. Initially, it is positioned one cell west of a north \guide. It moves diagonally south-west by alternatingly moving one field south and one field west until it hits a cell containing the west \guide whereafter it changes its movement direction to south-east, north-east, and finally north-west when hitting the respective \guides. When the \explorer arrives at the north \guide after it has completed the movement along the rectangle, it moves one cell north (in parallel with the north \mguide) and becomes a \mexplorer. \bigskip

\subsection{Analysis}

We denote by \emph{level} $d$ the set of all cells in distance $d$ from the origin. We say that a cell in distance $d$ is \emph{explored} when it is visited by an \explorer exploring level $d$. An \explorer is said to \emph{start} a rectangle search in level $d$ at time $t$ if it moves west from the cell $(0, d)$ (containing the north \guide) at time $t$. It \emph{finishes} a rectangle search in level $d$ at time $t$ if it moves onto cell $(0, d)$ (containing the north \guide) from the east at time $t$.
The \emph{start time} $\stime_d$ and \emph{finish time} $\ftime_d$ of a level
$d$ are given by the time when an \explorer starts or finishes a rectangle
search in level $d$, respectively.
An \explorer \emph{explores distance} $d$ or performs a \emph{rectangle search in level} $d$ at time $t$, if it has started a rectangle search in level $d$ at time $t' < t$ and has not yet finished it at time $t$. 

The design of \RectSearch ensures that regardless of which emission
scheme is used, the subset of \guides in every cardinal direction occupy a
contiguous segment of cells.
It also ensures the following two observations.

\begin{observation}
	\label{obs:search time}
  For any level $d > 0$, $\ftime_{d} - \stime_{d} = 8d$.
\end{observation}

\begin{observation}
  \label{obs:start times}
  $\stime_{d} - \stime_{d'} \geq d - d'$ for every $d \geq d'$.
\end{observation}
\begin{proof}
Let $e$ and $e'$ be the agents exploring levels $d$ and $d'$, respectively.
If $e = e'$, then the assertion holds trivially, so assume that $e \neq e'$.
When agent $e$ arrives to cell $(0, d')$ as a \mexplorer at some time $t'$,
agent $e'$ must have already started exploring level $d'$ as otherwise, agent
$e$ would have explored this level;
that is, $t' > \stime_{d'}$.
It will take agent $e$ at least $d - d'$ more time units to get to cell $(0,
d)$ and start exploring level $d$.
\end{proof}

\noindent
We are now ready to prove the following lemma essential for the
correctness of our algorithm.

\begin{lemma}
	\label{lem:same state different cell}
Outside the origin, no two agents of the same type occupy the same cell at
the same time.
\end{lemma}
\begin{proof}
First observe that no two \nguides or \nexplorers can ever be in the same cell
since they are emitted in different rounds and no agent can become \nguide or
\nexplorer again.
Combining Observations~\ref{obs:search time} and \ref{obs:start times} with
the fact that no two \nexplorers are emitted from the origin at the same time,
we conclude that two \mexplorers cannot occupy the same cell at the same time
and therefore, neither can two \explorers.
A similar argument establishes the assertion for \mguides.
\end{proof}

Notice that a \nexplorer and a \mexplorer (or a \nguide and a \mguide) may
share a cell, but this does not affect our algorithm since they are
distinguishable.

For the sake of a clearer run-time analysis, we analyze \RectSearch employing
an ideal emission scheme with emission function
$f_n(t) = t$,
i.e., a new search team is emitted from the origin every round.
We do not know how to implement such a scheme, but in Section~\ref{sec:almost
optimal emission scheme} we will describe an emission scheme with an almost
ideal emission function of
$f_n(t) = \Omega(t - \log n)$.

Consider some cell $c$ in level $d \leq n$.
Observe that under the assumption that $f_{n}(t) = t$, the $d^{\text{th}}$
emitted explorer will start exploring level $d$ at time $2 d$ unless some
previously emitted explorer already did so.
Therefore, cell $c$ is explored in time $\bigO (d)$ as promised.
So, it remains to prove that $c$ is explored in time $\bigO (d^{2} / n)$
for every cell $c$ in level $d > n$. 
The following lemma plays a key role in this regard.

\begin{lemma}
	\label{lem:distance between mexplorers}
Consider an \explorer when it finishes exploring level $d$ and becomes a
\mexplorer at time $\ftime_d$.
Then at that time, every other \mexplorer is at distance at least $d+8$ from
the origin.
\end{lemma}
\begin{proof}
Let $e$ be an \explorer that finishes exploring level $d$ and becomes a
\mexplorer at time $\ftime_d$ and consider some other \mexplorer $e'$ at that
time.
Let $d'$ be the last level explored by agent $e'$ prior to time $\ftime_d$.
Since at time $\stime_{d}$, agent $e'$ must have been at distance less than
$d$ from the origin and since every level larger than $d$ takes longer time to
explore than level $d$, it follows that $d' < d$.
Combining Observations~\ref{obs:search time} and \ref{obs:start times}, we
conclude that during the time interval $[\stime_{d'}, \ftime_{d}]$, agent $e'$
managed to complete the exploration of level $d'$ and move
$d - d' + 8 (d - d') \geq d - d' + 8$
steps outwards as a \mexplorer.
\end{proof}

\begin{corollary}
	\label{cor:distance between mexplorers}
The distance between any two \mexplorers is at least $8$.
\end{corollary}

Let $t_0$ be the first time at which there are no \nexplorers anymore and note
that $t_0 = \bigO(n)$.
Corollary~\ref{cor:distance between mexplorers} implies that once all
\nexplorers are gone, most of the explorers are busy exploring new cells.

\begin{observation}
	\label{obs:constant fraction exploring}
At any time $t > t_0$, at least $7/8$ of the exploring agents are \explorers.
\end{observation}

\begin{corollary}
	\label{cor:constant exploring rate}
At any time $t > t_0$, $\Omega(n)$ new cells are being explored.
\end{corollary}

\noindent
The main theorem of this section can now be established.

\begin{theorem}
	\label{thm:otpr rectsearch}
Employing an emission scheme with $f_n(t) = t$, \RectSearch locates the
treasure in time $\bigO(D + D^2/n)$.
\end{theorem}
\begin{proof}
The case of $D \leq n$ was already covered, so assume that $D > n$.
Since the \guides occupy a contiguous segment of cells, it follows that no
cell in level larger than $D + n / 5 < 2 D$ is explored before the exploration
of level $D$ is completed.
The number of cells in the first $2 D$ levels is $\bigO (D^{2})$.
The assertion is completed by Corollary~\ref{cor:constant exploring rate}
ensuring that all these cells will be explored by time
$t_{0} + \bigO (D^{2} / n) = \bigO (D^{2} / n)$.
\end{proof}

\section{An Almost Optimal Emission Scheme}
\label{sec:almost optimal emission scheme}

In this section, we introduce the emission scheme \emph{\PSTA} that \whp guarantees an emission function of $f_n(t) = \Omega(t - \log n)$. Plugging \PSTA into the \RectSearch strategy described in the previous section can be viewed as if after time $k \log n$, a new search team is emitted every constant number of rounds. By Theorem~\ref{thm:otpr rectsearch} this results in an almost optimal run-time of $\bigO(D + D^2 /n + \log n)$ \whp.

In Section~\ref{sec:optimal rectangle search} we describe the search strategy \GeomSearch, that when combined with \RectSearch yields an optimal run-time of $\bigO(D + D^2/n)$. \medskip

\noindent The main goal of this section is to establish the following theorem.

\begin{theorem}
	\label{thm: almost optimal}
	Employing the \PSTA emission scheme, \RectSearch locates the treasure in time $\bigO( D + D^2 / n + \log n)$ \whp.
\end{theorem}

Our first goal is to describe the process \FS, where $n$ agents spread out along the east ray $R$ consisting of the cells $(x, 0)$ for $x \in \mathbb N_{>0}$ such that eventually for a prefix of the ray $R$ every cell contains a single agent.

The idea is that in every round each agent tosses a fair coin and moves east with probability $1/2$ or stays put in its current cell otherwise. When an agent senses that it is the only agent in its cell, it marks itself as \textit{ready} and stops moving. To avoid cells being left without any agents, the agents check before moving that not all agents in the corresponding cell decided to move and stays put if that is the case. Furthermore, when an agent walks onto a cell with an agent already marked as ready, it moves one cell further east regardless. We refer to a cell that has once been visited by a ready agent as a ready cell.

\begin{lemma}
	\label{lemma: ready cells}
  For every positive integer $s \leq 6n$ and for every constant $k$, the first $s/6$ cells of the ray $R$ are ready after $s + k \log n$ rounds \whp.
\end{lemma}

\begin{proof}
Let $X_a$ be the random variable that counts the number of moves a not ready agent $a$ made towards east. Since the probabilities of $a$ moving forward are not independent, we study a weaker process where the number of movements for $a$ is dominated by $X_a$. Let $c$ be the cell occupied by $a$. Note that if $a$ is the only not ready agent occupying $c$, then it moves forward with probability $1$. Let $a_0, \ldots, a_z$ be the set of agents occupying $c$, assume that $z \geq 2$ and let $a = a_i$. In the weaker process, $a_i$ only moves forward if agent $a_{i + 1 \pmod{z}}$ decides to stay put. In other words, the probability of $a$ successfully moving forward is $1 / 4$.

Let $X'_a$ be the random variable that counts the number of moves $a$ made east in the weaker process. Now assume that $a$ is not ready in round $s + k \log n$. Then $E[X'_a] \geq 1/4 \cdot (s + k \log n)$. By applying a Chernoff bound we get that
\[
P(X'_a < 2/3 \cdot E[X'_a]) < P(X'_a < 2/3 \cdot 1/4 \cdot (s + k \log n)) < \bigO(n^{-k}) \enspace .
\]
Since $X_a \geq X'_a$, for any agent $a$ that that is not ready on round $s + k \log n$, the distance to the origin is at least $s/6$ \whp.
\end{proof}

Intuitively, the aforementioned process can be seen as parallel leader election. Since we want to describe an efficient emission scheme, it remains to show how the process can be used to quickly emit search teams consisting of five agents from the origin. To enable the \FS procedure to elect five different kind of agents per search team, we dedicate every fifth cell to a specific kind of agent. As an example, every cell in distance $d \equiv 1 \pmod 5$ is dedicated to an \explorer. After an \explorer is alone in a cell using the \FS procedure described above, it collects its search team in the following manner: it first takes one step east where a leader election for a \guide takes place. If the corresponding \guide cell is occupied by a \guide that is marked ready, they both move outwards to collect the next \guide. Otherwise, the \explorer waits until the leader election is over. After the \explorer (accompanied by the collected \guides) collected all four \guides needed for the search team, the team walks to the origin from where it will then be emitted into the four cardinal directions. We refer to the \FS protocol combined with the collection of the agents as \PSTA.

To prevent two search teams from entering the origin at the same time, we keep track of the innermost search team. In particular, we assign a flag to the agents in the leftmost cell that has not been collected. Every time an agent successfully moves to east by a coin toss in the leader election procedure, the flag is turned off on the respective agent. During the collection phase after all the agents have been collected, the \explorer makes an additional move east to the cell dedicated to elect the next \explorer and turns on the flag for all agents occupying this cell. Then the \explorer turns its flag off and moves back to its \guides. While the \explorer passes flag on, the \guides in its search team wait for its return. An \explorer executes the collection process only if it has the innermost flag activated. The use of the grid by the \PSTA protocol is illustrated in Figure~\ref{fig: psta}.

\begin{figure}	
	\centering
	\includegraphics{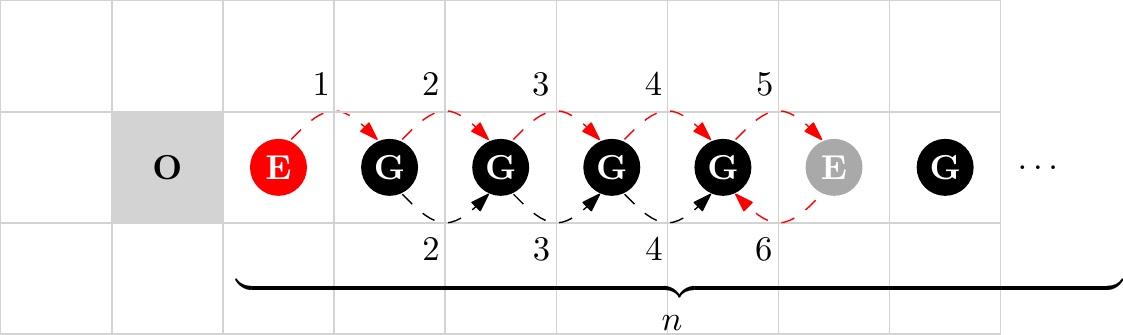}
	\caption{First, the $n$ agents executing the \PSTA protocol form a ray of single agents in cells $(0, 1), \ldots (0, n)$. After the innermost \explorer $e$ in cell $(0, 1)$ (denoted with red circle) is ready, it starts collecting its search team. Assuming that all the agents to join its search team are ready, after five rounds, $e$ enters the cell dedicated to the \explorer $e'$ of the next search team. Then the innermost flag of $e'$ is turned on and in the sixth round $e$ returns to cell $(0, 5)$, where its search team is waiting for it.}
	\label{fig: psta}
\end{figure}

\begin{corollary}
	Assume that $n$ agents start executing \PSTA protocol in round $0$. Then in round $8s + k \log n$, at least $ \lfloor \min\{s, n\}/5 \rfloor$ search teams have entered the origin \whp.
	\label{cor:parallel leader election}
\end{corollary}
\begin{proof}
	By Lemma~\ref{lemma: ready cells}, the first $s$ cells are ready \whp at time $t_0 = 6s + k \log n$, which indicates that the agents in these cells are ready to perform their collection process latest at time $t_0$. In addition, in each round after $t_0$, the \explorer flagged as innermost and occupying one of the first $s$ cells moves east. Therefore, latest in round $s + 6s + k \log n = r'$, all \explorers with distance of at most $s$ from the origin have been flagged as innermost. Since the \explorer of the successive search team is flagged  after the previously innermost search team is ready to walk back to the origin, at least $\lfloor s/5 \rfloor$ search teams have started walking at round $r'$. Furthermore, every search team has to walk at most $s$ steps towards the origin and therefore, in round $8s + k \log n$ all of the $\lfloor \min\{s, n\}/5 \rfloor$ search teams, consisting of agents distance at most $s$, have reached the origin \whp.
\end{proof}

By Corollary~\ref{cor:parallel leader election}, the emission function $f_n (t)$ provided by the \PSTA protocol satisfies $f_n (t) = \Omega(t - \log n)$ and therefore, Theorem~\ref{thm: almost optimal} follows.

\section{Optimal Rectangle Search}
\label{sec:optimal rectangle search}

In this section, we will present the search strategy \emph{\HybridSearch} that locates the treasure with optimal run-time of $\bigO(D + D^2/n)$. This is achieved by, combining \RectSearch employing the \PSTA with the randomized search strategy \GeomSearch that is fast only if the treasure is close to the origin. \medskip

The search strategy \GeomSearch is suited to locate the treasure very quickly if it is located close to the origin, more precisely if $D \leq \log (n) / 2$. Initially, each of the $n$ agents chooses uniformly at random one of the four quarter-planes that it will be searching. We will explain the strategy exemplary for an agent ``responsible'' for the north-east quarter-plane. The other three types operate analogously in their respective quarter-plane.

Initially, the agent moves one cell to the east. From then on, it moves a geometrically distributed number of steps east following which it moves a geometrically distributed number of steps to the north. More precisely, with probability $1/2$ the agent moves further and otherwise stops walking in the current direction. Both these processes can be realized in our model by having two state transitions where one of them moves the agent further while the other one ends the current walk. Either of the two transitions is chosen uniformly at random and a walk of geometrically distributed length is obtained.

\begin{lemma}
If $D \leq \log (n) / 2$, then \GeomSearch locates the treasure in time
$\bigO(D)$ \whp.
\end{lemma}
\begin{proof}
Consider some cell $c$ at distance $d \leq \log (n) / 2$ from the origin and
fix some agent $a$.
Let $X_a$ be a random variable that captures the length of the walk of agent
$a$ and observe that $X_a$ obeys a negative binomial distribution so that
\[
P(X_a = k) = (k+1) \cdot 2^{-(k+2)} \enspace .
\]
Recalling that $a$ has already moved one step, we conclude that the
probability that $a$ moves up to distance $d$ is
\[
P(X_a = d-1) = d \cdot 2^{-d-1} \geq 2^{-d-1} = \Omega (1 / \sqrt{n}) \enspace .
\]
Since all cells at distance $d$ from the root have the same probability of
being explored by $a$ and since there are $\bigO (\log n)$ such cells, it
follows that $a$ explores cell $c$ with probability at least
$\Omega \left( \frac{1}{\sqrt{n} \log n} \right)$.
Therefore, the probability that none of the agents explores cell $c$ is at most
\[
\left( 1 - \Omega \left( \frac{1}{\sqrt{n} \log n} \right) \right)^{n} <
e^{-\Omega \left( \frac{\sqrt{n}}{\log n} \right)} \, .
\]
The assertion follows.
\end{proof}

We can now combine the two search strategies \GeomSearch, which is optimal for $D \leq \log(n)/2$, and \RectSearch employing \PSTA, which is optimal for $D = \Omega(\log n)$, into the \HybridSearch strategy as follows.

At the beginning of the execution, each agent tosses a fair coin to decide whether it participates in \RectSearch or \GeomSearch. Let $n_r$ and $n_g$ be the number of agents participating in \RectSearch and \GeomSearch, respectively and observe that $n_r,n_g \geq n/3$ \whp Then the agents enter according states so that they do not interfere with each other anymore. One group executes \GeomSearch and locates the treasure \whp in time $\bigO(D)$ if $D \leq \log(n)/2$ and the other group executes \RectSearch locates the treasure \whp in time $\bigO(D+D^2/n)$ if $D = \Omega(\log n)$, thereby establishing the main theorem.

\begin{theorem}
	\HybridSearch locates the treasure in time $\bigO(D+D^2/n)$ \whp.
\end{theorem}

\section{Conclusions.}

In this paper, we establish a tight $\bigO (D + D^{2} / n)$ upper bound on the
time it takes for $n$ FSM-controlled agents to locate a treasure hidden at
distance $D$ from the origin.
Combined with the lower bound of Feinerman et al., our result demonstrates
that by allowing the agents to use a very primitive mean of communication, one
can get rid of the requirement for a super-constant memory and an
approximation of $n$.
This last observation may come as good news to anyone interested in studying
the communication between collaborating entities.

The aforementioned upper bound is based on a randomized algorithm that locates
the treasure in finite time with probability $1$ (i.e., a Las Vegas algorithm)
and in time $\bigO (D + D^{2} / n)$ \whp.
It is not difficult to extend our analysis showing that the run-time of our
algorithm holds also in expectation.

As mentioned in Section~\ref{section:Model}, we assume that the agents operate
in a synchronous environment.
Given the natural motivation for our model (studying ant colonies), it may be
desirable to lift this assumption.
Although this issue is beyond the scope of the current extended abstract, note
that our algorithm can be adapted to work in a fully asynchronous environment
as long as the model is changed so that an agent's communication range is
extended to include its neighboring cells on top of its own cell.
Whether the ANTS problem can be solved within the same asymptotic time bounds
in an asynchronous environment without extending the communication range is an
interesting open question.

One may wonder if locating the treasure is indeed the ``right'' goal for
agents whose navigation capabilities are so weak (cannot even store their own
coordinates):
Can an agent locating the treasure find its way back to the origin?
Can we guarantee that all agents eventually return to the origin?
To that end, notice that our algorithm can be modified to ensure a positive
answer for these two questions.
In fact, we can guarantee that the treasure finder returns to the origin in time
$\bigO (D + D^{2} / n)$
\whp and that all agents return to the origin in time
$\bigO (D + D^{2} / n + \log n)$
\whp. Getting rid of the extra logarithmic term in the latter bound seems to be an
interesting challenge.
In any case, the treatment of these secondary goals is deferred to the full
version.

Finally, it is important to point out that our algorithm is not inspired by
any observations regarding the real behavior of ants.
(In fact, we will be surprised if the \RectSearch strategy that lies at
the heart of our algorithm fits an exploration pattern used by real ants.)
As such, we do not claim that our results explain any natural phenomenon, but
rather attempt to advance the understanding of the power and limitations of a
basic nature-inspired model.

\clearpage
\renewcommand{\thepage}{}

\bibliographystyle{abbrv}
\bibliography{references}

\end{document}